\documentclass[conference]{IEEEtran}
\makeatletter
\def\ps@headings{%
\def\@oddhead{\mbox{}\scriptsize\rightmark \hfil \thepage}%
\def\@evenhead{\scriptsize\thepage \hfil \leftmark\mbox{}}%
\def\@oddfoot{}%
\def\@evenfoot{}}
\makeatother
\usepackage[normalem]{ulem}

\usepackage[cmex10,intlimits]{amsmath} 
\usepackage{amssymb,amsthm}
\usepackage{graphicx,gastex}
\usepackage[loose]{subfigure}
\usepackage{cite}
\usepackage{enumerate}
\usepackage{xcolor}
\usepackage{tikz} 
\pgfrealjobname{cache_main}

\usetikzlibrary{positioning}
\usetikzlibrary{calc}
\usepackage{pgfplots}

\newcommand{\fm}[1]{\scriptsize\mbox{\ensuremath{#1}}}
\newcommand{\ft}[1]{\scriptsize #1}

\usepackage{varioref}
\usepackage[breaklinks=true,colorlinks=true,linkcolor=black,anchorcolor=black,citecolor=black,filecolor=black,menucolor=black,urlcolor=black,plainpages=false,bookmarksopen=true,bookmarksnumbered=true]{hyperref}

\newtheorem{theorem}{Theorem}
\newtheorem{lemma}{Lemma}
\newtheorem{corollary}{Corollary}

\newcommand{\ie}{{\it i.e.},\ }

\newcommand{\cf}{{\it cf.} }

\newcommand{\EE}{\ensuremath{\mathbb{E}}}

\newcommand{\RRP}{\ensuremath{\mathbb{R}_+}}

\newcommand{\stleq}{\ensuremath{\leq_{\mathrm{st}}}}

\newcommand{\FF}{\ensuremath{\mathbb{F}}}
\newcommand{\II}{\ensuremath{\mathbf{I}}}

\newcommand{\ZZ}{\ensuremath{\mathbf{Z}}}

\newcommand{\dmax}{\ensuremath{\delta_{\mathrm{max}}}}

\newcommand{\Gdeff}{\ensuremath{G_0}}
\newcommand{\Gmin}{\ensuremath{G_{\mathrm{min}}}}
\newcommand{\Gmax}{\ensuremath{G_{\mathrm{max}}}}

\newcommand{\dd}{\ensuremath{d}}

\begin{document}

\title{Coding for Caches in the Plane}

\author{
\IEEEauthorblockN{
Eitan Altman\IEEEauthorrefmark{1},
Konstantin Avrachenkov\IEEEauthorrefmark{1} and
Jasper Goseling\IEEEauthorrefmark{2}\IEEEauthorrefmark{3}
}
\IEEEauthorblockA{
 \IEEEauthorrefmark{1}
INRIA Sophia Antipolis, France\\
}
\IEEEauthorblockA{
 \IEEEauthorrefmark{2}
 University of Twente, The Netherlands\\
}
\IEEEauthorblockA{
 \IEEEauthorrefmark{3}
 Delft University of Technology, The Netherlands\\
 eitan.altman@sophia.inria.fr, k.avrachenkov@sophia.inria.fr, j.goseling@utwente.nl
}
}

\maketitle

%
%
%
\begin{abstract}
We consider wireless caches located in the plane according to general point process and specialize the results for the homogeneous Poisson process. A large data file is stored at the caches, which have limited storage capabilities. Hence, they can only store parts of the data. Clients can contact the caches to retrieve the data.  We compare the expected cost of obtaining the complete data under uncoded as well as coded data allocation strategies. It is shown that for the general class of cost measures where the cost of retrieving data is increasing with the distance between client and caches, coded allocation outperforms uncoded allocation. The improvement offered by coding is quantified for two more specific classes of performance measures. Finally, our results are validated by computing the costs of the allocation strategies for the case that caches coincide with currently deployed mobile base stations.
\end{abstract}



%
%
%
\section{Introduction}
Consider base stations located in the plane that serve as caches of a big piece of data. Clients are interested in obtaining the complete data. The cost of obtaining data from a particular cache is increasing with the distance between the client and that cache. This cost can reflect, for instance, the time required to obtain the data, the transmit power that needs to be used at the cache to satisfy certain QoS constraints, or the probability with which the client fails in obtaining the data.

\begin{figure}
%
%
%

\centering
\begin{tikzpicture}
\tikzstyle{boxa}=[rectangle,draw,minimum width=2mm,minimum height=2mm]
\tikzstyle{boxb}=[rounded corners,draw,minimum width=20mm,minimum height=4mm]
\tikzstyle{->}=[-latex]

\path (0,0) -- (0,3);

\node[boxa] (client) at (0,0) {\ft{client}};
\node[boxb] (closeb) at (-.8,1.5) {\fm{x_2}};
\node[boxb] at (2.3,1.7) {\fm{x_1}};
\node[boxb] (closea) at (1.8,-0) {\fm{x_1}};
\node[boxb] at (-1.3,-1) {\fm{x_1}};
\node[boxb] at (3,-1.5) {\fm{x_2}};

\draw[latex-] (client) -- (closeb);
\draw[latex-] (client) -- (closea);
\end{tikzpicture}
\caption{Uncoded data allocation}
\label{fig:intropart}

\centering
\begin{tikzpicture}
\tikzstyle{boxa}=[rectangle,draw,minimum width=2mm,minimum height=2mm]
\tikzstyle{boxb}=[rounded corners,draw,minimum width=20mm,minimum height=4mm]
\tikzstyle{->}=[-latex]

\path (0,0) -- (0,3);

\node[boxa] (client) at (0,0) {\ft{client}};
\node[boxb] at (-.8,1.5) {\fm{\alpha_3x_1+\beta_3x_2}};
\node[boxb] at (2.3,1.7) {\fm{\alpha_5x_1+\beta_5x_2}};
\node[boxb] (closex) at (1.8,-0) {\fm{\alpha_1x_1+\beta_1x_2}};
\node[boxb] (closey) at (-1.3,-1) {\fm{\alpha_2x_1+\beta_2x_2}};
\node[boxb] at (3,-1.5) {\fm{\alpha_4x_1+\beta_4x_2}};

\draw[latex-] (client) -- (closex);
\draw[latex-] (client) -- (closey);
\end{tikzpicture}
\caption{Coded data allocation}
\label{fig:introcoded}
\end{figure}

The capacity of the caches is limited, hence it is not possible to store the complete data in each cache. In this paper we study \emph{allocation strategies} for the caches, \ie strategies that decide what data to store in which cache. In particular, we will study two strategies that we will introduce next by means of an example. Suppose that the data consists of two packets, $x_1$ and $x_2$, of equal size. Assume, moreover, that the capacity of a cache equals the size of a single packet.

One particular strategy that we will analyze is to store in each cache either $x_1$ or $x_2$. More precisely, for each cache we choose at random, independent of the other caches, to store $x_1$ (with probability $1/2$) or $x_2$ (with probability $1/2$). Now, in order to obtain all data, a client would need to contact two caches. In fact, in order to minize cost it will request $x_1$ from the nearest cache that has $x_1$ and $x_2$ from the nearest cache that has $x_2$. We will refer to this strategy, illustrated in  Figure~\ref{fig:intropart}, as the uncoded strategy.

The second strategy that we consider is based on caching coded data. In particular, we store at each cache a random linear combination of $x_1$ and $x_2$. Now, with high probability, the linear combinations of any set of two caches will form a full rank system. Therefore, the client can contact the two closest caches and obtain sufficient information to retrieve the original data. This is illustrated in Figure~\ref{fig:introcoded}. It is suggested by Figures~\ref{fig:intropart} and~\ref{fig:introcoded} that the coded strategy outperforms the partitioning strategy in the sense that the distances between client and contacted caches are reduced. One of the contributions of the current work is to proof this result rigorously for a broad class of performance measures.

In addition, we will be interested in two more specific performance measures. The first is the cache hit rate, \ie the probability that the client can retrieve the data from the caches. The second performance measure is the expected total cost of retrieving the data, where total cost is the sum of costs of obtaining individual fragments. The cost of obtaining a fragment is increasing in the distance between client and cache.

Applications of linear network coding for distributed storage were studied in~\cite{dimakis2010network}, where it has been demonstrated that repair bandwidth can be significantly reduced.
In the current work there is no notion of repair. Instead, coding is used to bring the data `closer' to the client.
The use of  coding was also explored in~\cite{dimakis2005ubiquitous} where it was shown how to efficiently allocate the data at  caches with the aim of ensuring that any sufficiently large subset of caches can provide the complete data.
The difference with the current work is that we are taking the geometry of the deployment of the caches into account.
In~\cite{maddah2012fundamental} coding strategies for networks of caches are presented, where each user has access to a single cache and a direct link to the source.
It is demonstrated how coding helps to reduce the load on the link between the caches and the source.
Note that we assume that different transmissions from caches to the clients are orthogonal, for instance by separating them in time or frequency.
In~\cite{6225433} the impact of non-orthogonal transmissions is considered and scaling results are derived on the best achievable transmission rates.
Systems of caches can be classified according to the amount of coordination between the caches. In~\cite{rosensweig09infocom} an approach with implicit coordination is proposed. Networks of caches are notoriously difficult to analyze. Only some very particular topologies
and caching strategies (see \cite{Fofack:2012} and references therein) or approximations \cite{Rosensweig2010},\cite{Che:2006} have been studied.
In a recent work \cite{Rosensweig2013} ergodicity of cache networks has been investigated. Using continuous
geometrical constraints on cache placement instead of combinatorial constraints allows us to obtain exact analytical results.

Other work on caching in wireless networks is, for instance, \cite{nuggehalli2003energy, jin2005content, yin2006supporting}. In \cite{nuggehalli2003energy} the authors analyze the trade-off between energy consumption and the retrieval delay of data from the caches.
In \cite{jin2005content}, the authors consider the optimal number of replicas of data such that the distance between a requesting node and the nearest replica is minimized.
Data sharing among multiple caches  such that the bandwidth consumption and the data retrieval delay are minimal is considered in  \cite{yin2006supporting}. None of~\cite{nuggehalli2003energy, jin2005content, yin2006supporting} are considering coded caching strategies.

We would like to emphasize that none of the above mentioned works considered continuous geometric constraints on cache placement.
Thus, to the best of our knowledge, the current paper is the first work on the analysis of spatial caching with stochastic geometry.

We would also like to note that the results of the present work can be applied to Information Centric Networking (ICN). ICN
is a new paradigm for the network architecture where the data is addressed by its name or content directly rather than by its physical location.
Examples of the ICN architecture are CCN/NDN \cite{Jacobson:2009}, DONA \cite{Koponen:2007} and TRIAD \cite{Gritter:2001}. Our results can be useful for the design of the wireless networks with the ICN architecture.



The remainder of this paper is organized as follows. In Section~\ref{sec:model} we define the model and formulate and exact problem statement.
Our main theoretical results are given in Section~\ref{sec:results}. In Section~\ref{sec:numerics} we provide numerical examples. In particular,
we consider network topologies generated by spatial Poisson process and from a real wireless network. We conclude and give future research
directions in Section~\ref{sec:discussion}.

%
%
%
\section{Model and problem statement} \label{sec:model}

\subsection{System}
We consider caches that are positioned in the plane and enumerate the caches with respect to their distance
to the client starting with cache closest to the client. The indices of the caches are denoted as $I_1,I_2,\dots$ and  we denote by $D(I)$ the distance between the client and the $I$-th cache. For a vector $\II=(I_1,I_2,\dots)$, let $D(\II)=(D(I_1),D(I_2),\dots)$. In the remainder we consider two scenarios: i) the caches are placed at arbitrary positions and we derive results that hold for any such placement, ii) caches are positioned according to a homogeneous Poisson process (HPP) with intensity $\lambda$. In the second scenario $D(I)$ is a random variable and we will derive results that hold in expectation.


The data $M$ consists of $k$ symbols, where each symbol is selected uniformly at random from~$\{1,\dots,q\}$, \ie
\begin{equation}
 M = \left(M_1,\dots,M_k \right).
\end{equation}
 A particular instance of such a $q$-ary symbol is that of a tuple of many bits. Each cache has the capacity to store one $q$-ary symbol.
There are various approaches to placing the message in the cache. Analyzing these placement strategies is the main problem studied in this paper. We introduce the strategies in Section~\ref{sec:strategies}.

A client, located at an arbitrary location in the plane, needs to retrieve $M$ by combining information from various caches. We assume that the client has full knowledge of the locations of the caches as well as the placement strategy that has been used. 
Clients can request data from caches. The cost of obtaining the data from caches is increasing with the distance between the client and these caches. A client will have to contact several caches to obtain the complete message. 

Our interest is in the expected cost, where the expectation is over the randomness in the placement strategy as well as over the location of the caches in case of placement according to a Poisson process.

%

%
%
%
\subsection{Strategies} \label{sec:strategies}

We will consider two strategies.

\subsubsection{Uncoded}
In the first strategy that we consider we simply store one of the symbols $M_1,\dots,M_k$ in a cache.
This is done independently and at random for each cache by choosing for each cache a label uniformly at random from $\{1,\dots,k\}$.
 In a cache with label $t$ we store $M_t$.

In the delivery phase the client contacts for each $t\in\{1,\dots,k\}$ the closest cache with that label and retrieves $M_t$. We denote by $I^p_1<I^p_2< \dots< I^p_k$ the index of the cache from which the parts are obtained, \ie if $I^p_i=j$, then the $i$-th part is obtained from the $j$-th nearest cache. Let $\II^p=(I_1^p,\dots,I_k^p)$.

\subsubsection{Random linear coding}
The second strategy is based on random linear coding. We interpret $M_1,\dots,M_k$ as symbols over the field $\FF_{q}$. For each cache we draw coefficients $c_i$, $i=1,\dots,k$, independently and uniformly at random from $\FF_{q}$. In the cache we store
\begin{equation}
\sum_{i=1}^k c_i M_i.
\end{equation}
In the delivery phase we retrieve from the set of $k$ nearest caches that provide complete data. In particular, the first part is obtained from the first neighbor. If the second neighbor has a linear combination that is linearly independent from the part of the first neighbor it is retrieved. Otherwise, we check the content of the third neighbor, etc. Let $I^c_1<I^c_2< \dots< I^c_k$ denote indices of the neighbors from which the parts are obtained. Let $\II^c=(I_1^c,\dots,I_k^c)$.


%
%
%
\subsection{Performance measures} \label{sec:perfmeasures}

We have defined two caching strategies. In both strategies the client retrieves the data from exactly $k$ caches. We express performance of the system in terms of the expected cost of contacting these caches. The cost is dependent on the distance from which the data is obtained. Let $\delta_1,\dots,\delta_k$ denote these distances. Using the notation from Section~\ref{sec:strategies} we have $\delta_i=D(I^p_i)$ and $\delta_i=D(I^c_i)$ in the uncoded and coded strategy respectively.

In the remainder we will present results for a general cost measure $G:\RRP^k\to\RRP$, where we used $\RRP=[0,\infty)$. In addition we will give results for a two classes of more specific cost measures. These classes will be introduced in the subsequent subsections. First we give the constraints that we impose on all the cost measures that we consider. We assume that $G$ is an increasing function in the sense that $\delta_i\geq \tilde\delta_i$ for all $i=1,\dots,k$ implies that $G(\delta_1,\dots,\delta_k)\geq G(\tilde\delta_1,\dots,\tilde\delta_k)$. In addition $G$ is assumed to be bounded by $\Gmax$, \ie $0\leq G(\delta_1,\dots,\delta_k)\leq \Gmax$ for all $(\delta_1,\dots,\delta_k)\in\RRP^k$. The rationale of the upper bound is that if the caches that provide complete data are located too far away a client can obtain the data directly from the source
using some other technology.



\subsubsection{Energy consumption, waiting time, etc.} \label{ssec:w}
We introduce the cost function
\begin{equation} \label{eq:wcost}
W(\delta_1,\dots,\delta_k) = \sum_{i=1}^k \min\{\left(\delta_i\right)^a, \left(\dmax\right)^a\}, 
\end{equation}
where $a>0$ is a parameter.


This cost function can be related to various quantities of interest. For instance, it can denote the energy consumed by a cache in order to successfully deliver data to a client. Based on
\begin{equation}
 R = \frac{1}{2}\log\left(1+P\delta^{-a}\right),
\end{equation}
where $a$ is the path loss exponent of the wireless medium, we see that the transmit power $P$ required to transmit at a guaranteed minimum rate $R$ satisfies $P=(\exp(2R)-1)\delta^a$. Therefore, the energy required to deliver a symbol from the cache is of the form~\eqref{eq:wcost}, with preconstant proportional to $\exp(2R)-1$. In this case $a$ denotes the path loss exponent, which in the plane satisfies $a\geq 2$.

Another example is given by the time required to obtain the data if TCP is used. The throughput of TCP is known to be inversely proportional to the round trip delay, \cf\cite{padhye1998modeling} or \cite{Altman:2005}. The later is given by the sum of the processing/queueing delay (which we assume zero) and the propogation delay, which is proportional to the distance $\delta$. Therefore, $W=C\delta$, \ie $a=1$, gives the time required to obtain data from a cache using TCP.


\subsubsection{Cache miss probability}
We consider the case that clients can only connect to caches within range $r$. If the complete data cannot be retrieved from the caches within this range, a cache miss occurs.
 If clients can connect to any cache within range $r$, but not outside, then $\EE[F(D(\II)]$, with
\begin{equation}
F(\delta_1,\dots,\delta_k) =
\begin{cases}
0,\quad &\text{if } \max_i\delta_i\leq r,\\
1,\quad &\text{otherwise,}
\end{cases}
\end{equation}
denotes the probability that data can not be obtained successfully. 

%
%
%
\section{Results} \label{sec:results}
\subsection{General cost measures}

Observe that the random linear coding strategy depends on the client to find a sufficient number of linearly independent combinations in caches that are not too far away. If $q$ is small there is a significant probability that the information of the next cache is not linearly independent. Hence, one might argue that for small values of $q$ there exists performance measures for which the partitioning strategy outperforms random linear coding. Our first result demonstrates that this is not true, \ie that random linear coding outperforms the partitioning strategy for any value of $q$ and any cost function.
\begin{theorem}  \label{th:codingbetter}
Consider an arbitrary placement of caches, a bounded increasing function $G$ and any $q$. Then
\begin{equation*}
\EE\left[ G\left(D(\II^c)\right) \right] \leq \EE\left[ G\left(D(\II^p)\right) \right],
\end{equation*}
\ie coding always outperforms partitioning.
\end{theorem}
The proof of the above theorem is given in Appendix~\ref{app:codingbetter}.
Note, that the above result in particular implies that coding is better than uncoded strategy
for any realization in case of placement according to a spatial Poisson process.


Since less than $k$ caches can never provide the complete data, it is clear that the minimum cost that one can hope to achieve is given by the cost of contacting the nearest $k$ caches. For notation convenience, denote the minimum expected cost as
\begin{equation}
\Gmin= \EE\left[ G\left(D(1), \dots, D(k)\right) \right].
\end{equation}

Our next result provides a bound on the deviation of the cost of the coded strategy from $\Gmin$.
%
\begin{theorem} \label{th:upper}
The expected cost of the coded strategy is upper bounded as
\begin{equation}
\EE\left[ G\left(D(\II^c)\right) \right] \leq \Gmin + \Gdeff, 
\end{equation}
where $\Gdeff=\Gmax(1-(1-1/q)^k)$.
\end{theorem}
\begin{IEEEproof} 
Let $A$ denote the event that the first $k$ caches provide a full rank system. Then
\begin{multline}
\EE\left[ G\left(D(\II^c)\right) \right]  = \EE\left[ G\left(D(\II^c)\right) | A \right]P(A) \\ + \EE\left[ G\left(D(\II^c)\right) | \bar A \right] P(\bar A)
\end{multline}
and the result follows from the assumption that $G$ is bounded by $\Gmax$ and $P(\bar A)\leq 1-(1-1/q)^k$, \cf~\cite{ho2006random} or~\cite{ncfundamentals}.
\end{IEEEproof}
We will illustrate the above result for a specific performance measure below. In particular, we will show that the coded strategy is close to optimal.

\subsection{Waiting time}

Next, we consider specific instances of the cost. First we consider $W$ as defined in Section~\ref{ssec:w}. Before giving the results we will provide some definitions and results on the Gamma function. Let $\Gamma(k,x)$ denote the upper incomplete Gamma function, \ie
\begin{equation}
 \Gamma(k,x) = \int_x^\infty z^{k-1}e^{-z}dz. 
\end{equation}
Furthermore, we define $\Gamma(k)=\Gamma(k,0)$ and $\gamma(k,x) = \Gamma(k) - \Gamma(k,x)$. The use of $\Gamma(k,x)$ for Poisson processes stems from the fact that for a Poisson random variable $X$ with mean $\mu$ we have
\begin{equation}
P(X<n) = e^{-\mu}\sum_{i=0}^{n-1}\frac{\mu^i}{\Gamma(i+1)} = \frac{\Gamma(n,\mu)}{\Gamma(n)}.
\end{equation}
Therefore, the probability that $D_n$, the distance to the $n$-th nearest neigbor in an HPP of density $\lambda$ is at most $\delta$ is
\begin{equation}
P\left(D_n\leq \delta\right) = 1-\frac{\Gamma(n,\lambda\pi\delta^2)}{\Gamma(n)}.
\end{equation}

We are now ready to present our first result on $W$. The proof of the next theorem is given in Appendix~\ref{app:wmain}.
\begin{theorem} \label{th:wmain}
Let the caches be distributed in the plane as HPP with density $\lambda$.
Then, the expected costs of the partitioning and coded strategies are
\begin{align*}
W^p =&\ k\left(\frac{k}{\lambda\pi}\right)^{b-1}\gamma\left(b,\frac{\dd}{k}\right)
 + k\left(\frac{d}{\lambda\pi}\right)^{b-1}\Gamma\left(1,\frac{\dd}{k}\right), \\
W^c_{\mathrm{min}} =&\ \frac{\gamma(k+b,d) + bd^{b-1}\Gamma(k+1,d)-(b-1)d^{b}\Gamma(k,d)}{(\lambda\pi)^{b-1}b\Gamma(k)},
\end{align*}
where $b=a/2+1$ and $d=\lambda\pi\dmax^2$.
\end{theorem}

To get an idea about the nature of the above involved equations, a reader may check Figure~\ref{fig:WHPPSimTheo}
where a numerical example is presented.



From the observation that
\begin{equation}
\lim_{d\to\infty} d^c\Gamma(k,d) = 0,
\end{equation}
for any $c$,
we immediately obtain the following corollary which provides limiting expressions for $W^p$ and $W^c_{\mathrm{min}}$ for $\dmax\to\infty$.
\begin{corollary} \label{cor:w}
Let $\bar W^p = \lim_{\dmax\to\infty} W^p$ and $\bar W^c_{\mathrm{min}} = \lim_{\dmax\to\infty} W^c_{\mathrm{min}}$. We have
\begin{align}
\bar W^p =& k\left(\frac{k}{\lambda\pi}\right)^{a/2}\Gamma\left(1+\frac{a}{2}\right), \\
\bar W^c_{\mathrm{min}} &=  \frac{\Gamma(k+1+\frac{a}{2})}{(1+\frac{a}{2})(\lambda\pi)^{a/2}\Gamma(k)}.
\end{align}
\end{corollary}

Next we turn our attention to the cost benefit of coding over partitioning, defined as the ratio $\bar W^p/\bar W^c_{\mathrm{min}}$.
Clearly, it is an optimistic prediction. However, in the numerical examples section we shall demonstrate that the actual gain is
typically close to this optimistic prediction. It follows readily from Corollary~\ref{cor:w} that
\begin{equation}
\frac{\bar W^p}{\bar W^c_{\mathrm{min}}}  = \left(1+\frac{a}{2}\right) k^{1+a/2}B\left(k,1+\frac{a}{2}\right),
\end{equation}
where $B(k,1+\frac{a}{2})$ denotes the Beta function.

\begin{theorem} \label{th:benefit}
The benefit of coding over partitioning, $\bar W^p/\bar W^c_{\mathrm{min}}$, is increasing in $k$. Moreover
\begin{equation}
\lim_{k\to\infty} \frac{\bar W^p}{\bar W^c_{\mathrm{min}}} = \left(\frac{a}{2}+1\right)\Gamma\left(\frac{a}{2}+1\right).
\end{equation}
Hence
\begin{equation}
1 \leq \frac{W^p}{W^c} \leq \left(\frac{a}{2}+1\right)\Gamma\left(\frac{a}{2}+1\right),
\end{equation}
with equality on the LHS iff $k=1$.
\end{theorem}
\begin{IEEEproof}
We start with proving that $\bar W^p/\bar W^c_{\mathrm{min}}$, is increasing in $k$.
Let $b=1+\frac{a}{2}$. We have
\begin{equation}
 \frac{\partial}{\partial k}\frac{\bar W^p}{\bar W^c_{\mathrm{min}}} = b k^b B(k,b)\left[ \frac{b}{k} + \psi(k) - \psi(k+b)\right],
\end{equation}
where $\psi(x) = \int_0^\infty \left( \frac{e^{-t}}{t} - \frac{e^{-xt}}{1-e^{-t}}\right)dt$ is the digamma function~\cite{abramowitz1974handbook}. We need to show that
\begin{equation}
 \int_0^\infty e^{-kt} \frac{1-e^{-bt}}{1-e^{-t}}dt \leq \frac{b}{k},
\end{equation}
but this follows directly from the observation that $\frac{1-e^{-bt}}{1-e^{-t}}\leq b$.

The limiting expression follows from an application of Stirlings approximation.
\end{IEEEproof}
Note, that for $a=1,2,3$, the upper bound in Theorem~\ref{th:benefit} reduces to $\frac{3\sqrt{\pi}}{4}\approx 1.3$, $2$ and $\frac{15\sqrt{\pi}}{8}\approx 3.3$ respectively.

\subsection{Cache miss probability}
Next let us consider the cache miss probability.
\begin{theorem}\label{thm:outage}
The cache miss probability of the partitioning and coded strategy are
\begin{align}
F^p &= 1 - \left( 1 - e^{-\lambda\pi r^2/k} \right)^k, \\
F^c_\mathrm{min} &= \frac{\Gamma(k,\lambda\pi r^2)}{\Gamma(k)}.
\end{align}
\end{theorem}
\begin{IEEEproof}
For the coded strategy we have
\begin{equation}\label{eq:Fcmin}
F^c_\mathrm{min} = P(D(k)>r) = \frac{\Gamma(k,\lambda\pi r^2)}{\Gamma(k)},
\end{equation}
since the $k$-th nearest cache needs to be within distance $r$.

For the uncoded strategy we have
\begin{align}
F^p
&= 1 - P_{\lambda/k}(D_1\leq r)^k \\
&= 1 - \left( 1 - e^{-\lambda\pi r^2/k} \right)^k, 
\end{align}
which follows from the fact that each of the parts is found in the first neighbour in a (thinned) Poisson process of intensity $\lambda/k$.
\end{IEEEproof}

We note that the following asymptotics for equations in the statement of Theorem~\ref{thm:outage} take place
$$
F^p \approx 1 - \frac{(\lambda\pi r^2)^k}{k^k},
$$
$$
F^c_\mathrm{min} \approx 1- e^{-\lambda\pi r^2} \frac{(\lambda\pi r^2)^k}{k!},
$$
for large values of $k$. This indicates that when increasing $k$, the cache miss probability increases much faster for the uncoded
strategy than for the coded strategy. We shall illustrate this phenomenon with a specific numerical example in the next section.

%
%
%
\section{Evaluation} \label{sec:numerics}

Let us evaluate our theoretical results for caches placed according to
(a) Spatial Homogeneous Poisson Process and (b) real wireless network.
In order to compare these two settings we have chosen the intensity of the
Poisson process equal to the density of base stations in the real network.
In Figure~\ref{fig:HPPrealization} one can see a realization of the
Spatial Homogeneous Poisson Process with $\lambda = 1.8324 \times 10^{-5}$.
The value is rather small because we measure distances in metres.
In Figure~\ref{fig:WHPPSimTheo} we have plotted $W^p$, $W^c_{min}$ and $W^c_{min}+W_0$
given by formulae of Theorems~\ref{th:upper}~and~\ref{th:wmain}
as functions of the number of parts in which a file is divided for $q=2^8$,
$a=2$ and $\delta_{max}=700$. On this figure we have also plotted the averaged
cost functions for uncoded and coded strategies from the simulation of the system.
We have averaged over 500 realizations of the Poisson process and 100 realizations
of the part distribution. The simulation curve for the uncoded strategy follows
very closely the theoretical result. Of course, this is no surprise if the simulations
and theoretical derivations have been done correctly. What is more interesting is
that the simulation curve for the coded strategy follows very closely the lower
bound given by $W^c_{min}$. This can be explained by the fact that when the number
of parts $k$ is not large, the event of obtaining not a full rank system after
$k$ requests to caches is very small.

\begin{figure}
\begin{center}
\includegraphics[scale=0.49]{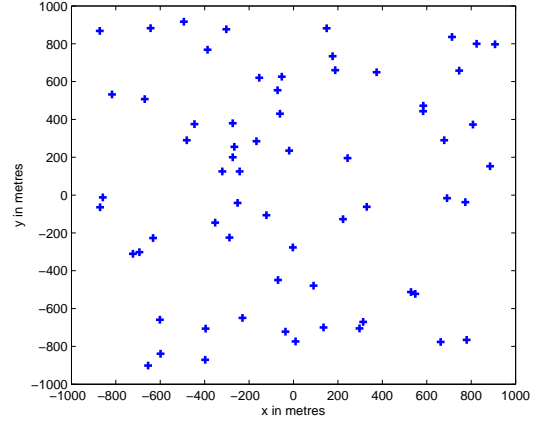}
\caption{A realization of the Spatial Homogeneous Poisson Process.}
\label{fig:HPPrealization}
\end{center}
\end{figure}

\begin{figure}
\begin{center}
\includegraphics[scale=0.49]{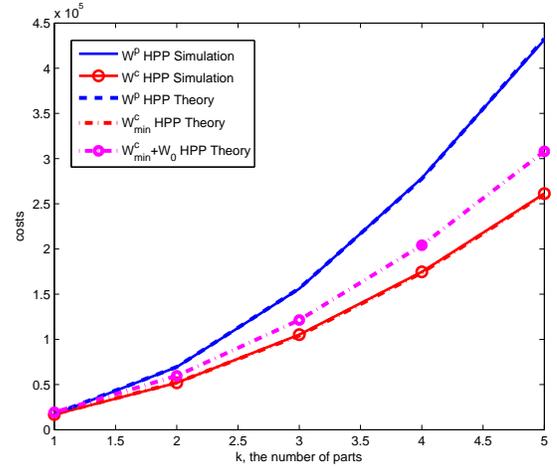}
\caption{Cost function, bounds and simulation results for Poisson process.}
\label{fig:WHPPSimTheo}
\end{center}
\end{figure}

Next we evaluate the performance of coded and uncoded strategies on the topology
of a real wireless network. We have taken the positions of 3G base stations provided
by the OpenMobileNetwork project \cite{OpenMobileNetwork}. The base stations are situated
in the area $1.95 \times 1.74 \ kms$ around the TU-Berlin campus. One can see the
positions of the base stations from the OpenMobileNetwork project in Figure~\ref{fig:realdata}.
We note that the base stations of the real network are more clustered because they are
typically situated along the roads. In Figure~\ref{fig:WRealDatavsHPP} we have plotted
the averaged simulation results for coded and uncoded strategies for the real network
topology and Poisson process with the intensity equal to the density of the base stations.
The averaging for the real network data has been done with respect to the user position
and part placement. In order to avoid boarder effect we have placed a user in a smaller
centered rectangular region with the lengths of the sides equal to halfs of the respective
sides.
We observe that the clustering actually decreases the cost and the decrease in the cost
of the uncoded strategy is larger than the decrease in the cost of the coded strategy.
A possible intuitive explanation for this phenomenon is that for uncoded strategy it is
difficult to find the last parts and clustering helps to reduce the marginal cost
of doing more trials at the end of the discovery process.

\begin{figure}
\begin{center}
\includegraphics[scale=0.49]{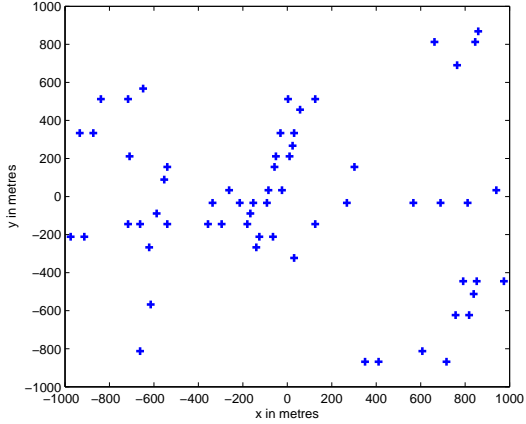}
\caption{Location of Base Stations from OpenMobileNetwork dataset.}
\label{fig:realdata}
\end{center}
\end{figure}


\begin{figure}
\begin{center}
\includegraphics[scale=0.49]{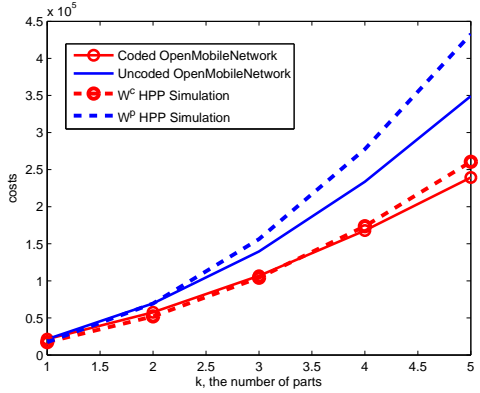}
\caption{Comparison OpenMobileNetwork versus Poisson process.}
\label{fig:WRealDatavsHPP}
\end{center}
\end{figure}

For the same $\lambda = 1.8324 \times 10^{-5}$ and $r=700$ we use the formulae from Theorem~\ref{thm:outage}
to plot the cache miss probabilities $F^p$ and $F^c_{min}$ in Figure~\ref{fig:Fpcmin}. Suppose one
wants the cache miss probability to be not larger than $10^{-2}$. Then, for the uncoded strategy
one should not divide the file into more than 5 parts, whereas for the coded strategy one can have
up to 17 parts. This highlights another benefit, flexibility, of the coded strategy for spatial
caching.

\begin{figure}
\begin{center}
\includegraphics[scale=0.49]{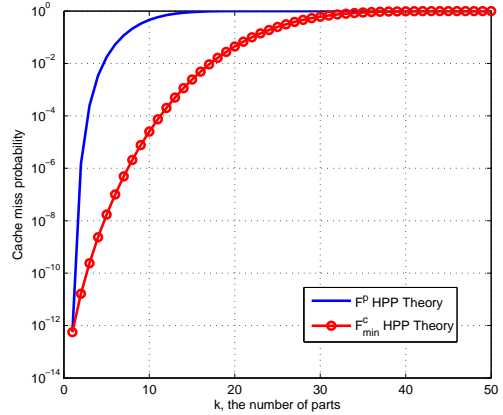}
\caption{Cache miss probability for Poisson process.}
\end{center}
\label{fig:Fpcmin}
\end{figure}

%
%
%

%
%
%
\section{Discussion} \label{sec:discussion}
In this paper we have analyzed allocation strategies for networks of caches. In particular, we have studied optimal deployment and allocation strategies for a single client that is requesting all data. In future work we will focus on the dynamic setting in which there are many clients that are requesting data over time. Hence, we obtain a spatial queueing model. One particular problem to study in this scenario is the  stability of such a system if we take into account the interference caused by multiple caches transmitting simultaneously. Another problem is on actual caching strategies if clients request only part of the data. In particular, if there are many different pieces of data with varying popularity, the challenge is to design distributed coded caching strategies that minimize overall cost, \ie maximize hit rate.

%
%
%
\section*{Acknowledgement}
Part of this work was done while J.~Goseling was visiting the MAESTRO group at INRIA Sophia Antipolis, France, in November 2011. The hospitality and support provided by INRIA are greatly acknowledged. This work was supported in part by the NWO grant 612.001.107.

\appendices

%
%
%
\section{Proof of Theorem~\ref{th:codingbetter}} \label{app:codingbetter}
Observe that Theorem~\ref{th:codingbetter} provides a stochastic comparison result. To simplify discussion in the remainder we introduce for random vectors $\ZZ=(Z_1,\dots,Z_k)$ and $\tilde\ZZ=(\tilde Z_1,\dots,\tilde Z_k)$ the notation
 $\ZZ \stleq \tilde\ZZ$,
to denote that $\EE\left[F\left(\ZZ\right)\right]\leq\EE\left[F\left(\tilde\ZZ\right)\right]$ for all bounded increasing $F$. We say that $\ZZ$ is less than $\tilde\ZZ$ in the usual stochastic ordering.

\begin{lemma} \label{lem:orderI}
$\II^c \leq_{\mathrm{st}} \II^p.$
\end{lemma}
\begin{IEEEproof}
We will make use of a result by Veinott~\cite{veinott1965optimal}, \cf~\cite[Theorem~3.3.7]{mullerstoyan} that states that the result
holds if
$I_1^c\stleq I_1^p$
and for $i=2,3,\dots,k,$
\begin{multline}
\left[ I_i^c \middle| I_1^c=x_1^c,\dots,I_{i-1}^c=x_{i-1}^c \right] \stleq \\
\left[ I_i^p \middle| I_1^p=x_1^p,\dots,I_{i-1}^p=x_{i-1}^p \right],
\end{multline}
whenever $x_j^c\leq x_j^p$ for all $j=1,\dots,i-1$.

We have $I_1^c=1$ and for $i=2,\dots,k$
\begin{equation}
 \left[ I_i^c \middle| I_1^c=x_1^c,\dots,I_{i-1}^c=x_{i-1}^c \right] = G_i + x_{i-1}^c,
\end{equation}
where $P(G_i=n)=(1-g_i)g_i^{n-1}$, $g_i=q^{i-1}/q^k$, \ie $G_i$ is geometrically distributed with parameter $g_i$. The above follows from the fact the coding coefficients are chosen independently for each cache. The probability that the vector of coefficients $(c_1,\dots,c_k)$ is linearly dependent of the $i-1$ vectors that have previously been obtained is $q^{i-1}/q^k$, where $q^k$ is the total number of possible vectors from which one is selected uniformly at random, and $q^{i-1}$ is the number of vectors that are linearly dependent.

For the partitioning strategy we have $I_1^p=1$ and
\begin{equation}
\left[ I_i^p \middle| I_1^p=x_1^p,\dots,I_{i-1}^p=x_{i-1}^p \right] = \tilde G_i + x_{i-1}^p,
\end{equation}
where $P(\tilde G_i=n)=(1-\tilde g_i)\tilde g_i^{n-1}$, $\tilde g_i=(i-1)/k$, which follows from the observation that there are $k$ different symbols that could be obtained, of which $i-1$ are not useful.

Now observe that $g_i<\tilde g_i$ for all $q$, all $k$, and all $i=2,\dots,k$. Therefore, $P(G_i\leq z)\geq P(\tilde G_i\leq z)$ for all $z$ and $G_i\stleq \tilde G_i$, \cf~\cite[Theorem~1.2.8]{mullerstoyan}. Finally, note that $G_i\stleq \tilde G_i$ implies that $G_i + x_{i-1}^c\stleq \tilde G_i + x_{i-1}^p$ whenever $x_{i-1}^c\leq x_{i-1}^p$.
\end{IEEEproof}
The proof of Theorem~\ref{th:codingbetter} follows directly from the above lemma.

%
%
%
\section{Proof of Theorem~\ref{th:wmain}} \label{app:wmain}

We first give some technical lemmas.

\begin{lemma} \label{lem:intgamma}
\begin{equation} \label{eq:lemintgamma}
\int_x^\infty z^c\Gamma(k,z)dz = \frac{\Gamma(k+c+1,x) - x^{1+c}\Gamma(k,x)}{c+1}.
\end{equation}
\end{lemma}
\begin{IEEEproof}
For $k=1$,~\eqref{eq:lemintgamma} follows from $\Gamma(1,z)=e^{-z}$ and an integration by parts. We proceed by induction over $k$. By repeatedly using the recurrence identity
\begin{equation}
\Gamma(k,x) = (k-1)\Gamma(k-1,x) + x^{k-1}e^{-x},
\end{equation}
which is easily shown using an integration by parts, and assuming that~\eqref{eq:lemintgamma} holds for $k-1$ we have
\begin{align}
\int_x^\infty z^c & \Gamma(k,z)dz
= \int_x^\infty z^{c+k-1}e^{-z}dz \\
=&\  \int_x^\infty \left[ z^{k+c-1}e^{-z} + (k-1)z^c\Gamma(k-1,z)\right]dz \\
=&\ \frac{(k+c)\Gamma(k+c,x) - x^{c+1}(k-1)\Gamma(k-1,x)}{c+1} \\
=&\ \frac{1}{c+1}\Big\{ \Gamma(k+c+1,x) \notag \\
   &\ - x^{k+c}e^{-x}
   - x^{c+1}\left[\Gamma(k,x) - x^{k-1}e^{-x}\right] \Big\} \\
=&\ \frac{\Gamma(k+c+1,x) - x^{1+c}\Gamma(k,x)}{c+1}.
\end{align}
\end{IEEEproof}

\begin{lemma} \label{lem:sumgammaincompl}
For $k\in\mathbb{N}$, $c\geq 0$ and $x\geq 0$
\begin{align}
\sum_{i=1}^k \frac{ \Gamma(i+c) }{\Gamma(i) } &= \frac{\Gamma(k+c+1)}{(c+1)\Gamma(k)}, \\
\sum_{i=1}^k \frac{ \Gamma(i+c,x) }{\Gamma(i) } &= \frac{\Gamma(k+c+1,x)-x^{c+1}\Gamma(k,x)}{(c+1)\Gamma(k)}, \\
\sum_{i=1}^k \frac{ \gamma(i+c,x) }{\Gamma(i) } &= \frac{\gamma(k+c+1,x)+x^{c+1}\Gamma(k,x)}{(c+1)\Gamma(k)},
\end{align}
\end{lemma}
\begin{IEEEproof}
We have
\begin{align}
\sum_{i=1}^k \frac{ \Gamma(i+c,x) }{\Gamma(i) }
=&\ \int_x^\infty \sum_{i=1}^k \frac{z^{c+i-1}}{\Gamma(i)}e^{-z}dz \\
=&\ \int_x^\infty z^c \sum_{j=0}^{k-1} \frac{z^{j}}{\Gamma(j+1)}e^{-z}dz \\
=&\ \int_x^\infty z^c \frac{\Gamma(k,z)}{\Gamma(k)}dz \\
=&\ \frac{\Gamma(k+c+1,x) - x^{c+1}\Gamma(k,x)}{(c+1)\Gamma(k)},
\end{align}
where the last equality follows from Lemma~\ref{lem:intgamma}. For $x=0$ this reduces to
\begin{equation}
 \sum_{i=1}^k \frac{ \Gamma(i+c) }{\Gamma(i) } = \frac{\Gamma(k+c+1)}{(c+1)\Gamma(k)}.
\end{equation}
\end{IEEEproof}

\begin{lemma} \label{lem:w}
\begin{equation}
\int_0^u x^a d\left(1-\frac{\Gamma(i,bx^2)}{\Gamma(i)}\right) = \frac{\gamma\left(\frac{a}{2}+i,bu^2\right)}{\Gamma(i)b^{a/2}}.
\end{equation}
\end{lemma}
\begin{proof}
The result follows directly from the definitions of $\Gamma(i,bx^2)$ and $\gamma(a/2+i,bu^2)$.
\end{proof}

Before giving the proofs for the two strategies, recall that $b=1+a/2$ and $d=\lambda\pi\dmax^2$. This leads to $\dmax^a=(d/\lambda/\pi)^{b-1}$ and $P(D(i)>\dmax)=\Gamma(i,\dd)/\Gamma(i)$.

\subsection{Uncoded strategy}
The first thing to note in the uncoded strategy is that for each of the labels $t\in\{1,\dots,k\}$ all caches with label $t$ form a homogeneous spatial Poisson process  with intensity $\lambda/k$. This follows from the fact that these caches form a thinned Poisson process~\cite{baccelli1}.

Now, for the uncoded strategy we have
\begin{align}
W^p
=& \sum_{i=1}^k \int_0^{\infty} \max\left\{\left(\delta_i\right)^a,\dmax^a\right\} d P(D(X_i^p)\leq\delta_i) \notag \\
=& k \int_0^{\dmax} \delta^a d\left( 1 - \frac{\Gamma(1,\frac{\lambda}{k}\pi\delta^2)}{\Gamma(1)} \right)
 + k\dmax^a\Gamma\left(1,\frac{d}{k}\right) \notag \\
=& k\left(\frac{k}{\lambda\pi}\right)^{a/2}\gamma\left(1+\frac{a}{2},\frac{\dd}{k}\right)
 + k\left(\frac{d}{\lambda\pi}\right)^{a/2}\Gamma\left(1,\frac{\dd}{k}\right), \notag
\end{align}
where we have used Lemma~\ref{lem:w}.

\subsection{Random linear coding strategy}
For the coded strategy we start with
\begin{equation}
W^c_{\mathrm{min}} = \sum_{i=1}^k \int_0^{\dmax} \delta^a d P(D(i)\leq\delta) + \dmax^a \sum_{i=1}^k \frac{\Gamma\left(i,\dd\right)}{\Gamma(i)},
\end{equation}
where we continue with analyzing each term seperately. First,
\begin{equation} \label{eq:wtemp1}
\dmax^a \sum_{i=1}^k \frac{\Gamma\left(i,\dd\right)}{\Gamma(i)} = \left(\frac{d}{\lambda\pi}\right)^{a/2} \frac{\Gamma(k+1,\dd)-\dd\Gamma(k,\dd)}{\Gamma(k)},
\end{equation}
by Lemma~\ref{lem:sumgammaincompl}. Next,
\begin{equation}
\int_0^{\dmax} \delta^a d P(D(i)\leq\delta) = \frac{\gamma\left(\frac{a}{2}+i,\dd\right)}{\Gamma(i)(\lambda\pi)^{a/2}},
\end{equation}
by Lemma~\ref{lem:w} and
\begin{equation} \label{eq:wtemp2}
\sum_{i=1}^k \frac{\gamma\left(\frac{a}{2}+i,\dd\right)}{\Gamma(i)(\lambda\pi)^{a/2}}
= \frac{\gamma(k+1+\frac{a}{2},\dd)+d^{a/2+1}\Gamma(k,\dd)}{(\lambda\pi)^{a/2}(\frac{a}{2}+1)\Gamma(k)},
\end{equation}
by Lemma~\ref{lem:sumgammaincompl}.
The result now follows from~\eqref{eq:wtemp1} and~\eqref{eq:wtemp2}.

%
%
%
\bibliographystyle{IEEEtran}
\bibliography{IEEEabrv,cache}

\end{document}